\documentclass[12pt,a4paper]{article}

\usepackage[utf8]{inputenc}
\usepackage[T1]{fontenc}
\usepackage{lmodern}

\usepackage{amsmath,amssymb,amsthm,amsfonts}
\usepackage{mathtools}
\usepackage{bm}

\usepackage{geometry}
\usepackage{graphicx}
\usepackage{float}
\usepackage{caption}
\usepackage{subcaption}

\usepackage{booktabs}
\usepackage{multirow}
\usepackage{array}

\usepackage{tikz}
\usetikzlibrary{arrows.meta,positioning,calc}

\usepackage{hyperref}
\hypersetup{
    colorlinks=true,
    linkcolor=blue,
    citecolor=blue,
    urlcolor=blue
}

\usepackage[numbers]{natbib}

\usepackage{setspace}
\newtheorem{theorem}{Theorem}[section]
\newtheorem{proposition}[theorem]{Proposition}
\newtheorem{lemma}[theorem]{Lemma}
\newtheorem{corollary}[theorem]{Corollary}

\theoremstyle{definition}

\newtheorem{remark}[theorem]{Remark}

\title{\textbf{
Strategic Inertia and Institutional Change:\\
A Behavioral Model of Price Reforms versus Action Deletion
}}

\author{
Madjid Eshaghi Gordji\thanks{Corresponding author: \texttt{meshaghi@semnan.ac.ir}}
\textsuperscript{1}
\and
Mohammadali Berahman\textsuperscript{1}
\and
Hasti Eshaghi\textsuperscript{2}
}

\date{}

\begin{document}

\maketitle

\begin{center}
\textsuperscript{1}
Faculty of Mathematics, Statistics and Computer Science, 
Semnan University, Semnan, Iran

\vspace{0.3cm}

\texttt{meshaghi@semnan.ac.ir}

\texttt{mohamadali\_berahman@semnan.ac.ir}

\vspace{0.5cm}

\textsuperscript{2}
Department of Mathematics, University of Tehran, Tehran, Iran

\vspace{0.3cm}

\texttt{hastieshaghi@gmail.com}
\end{center}

\vspace{1cm}

\begin{abstract}
Why do inefficient practices, technologies, or institutions persist even when superior alternatives are available? This paper introduces a quantal response equilibrium with status-quo bias (QRE-SB) in which each player incurs a fixed switching cost when deviating from an inherited default action. In a binary coordination game, we compare two policy interventions: a tax on the default action (price-only reform) versus deleting the default action entirely (ban). We prove that there exists a threshold tax below which the status quo persists and above which a transition occurs; notably, this threshold does not depend on the degree of bounded rationality. Deleting the default action always forces play to the superior equilibrium, irrespective of switching costs or rationality. Moreover, when the superior equilibrium is Pareto-dominant, deletion yields strictly higher expected welfare than any finite tax that leaves the old action feasible. Numerical simulations illustrate the theoretical predictions. The framework provides a formal foundation for the policy principle that sometimes you must ban, not just tax, with direct applications to climate policy, social media regulation, and international sanctions.
\end{abstract}

\noindent\textbf{Keywords:}
Bounded rationality, quantal response, switching costs, status quo bias, institutional change, action deletion.

\vspace{0.3cm}



\section{Introduction}

Why do inefficient practices, technologies, or institutions persist even when superior alternatives are available? Economists have long recognized that coordination failures and path dependence can lock in suboptimal outcomes \citep{david1985,arthur1989,north1990}. More recently, behavioral economists have documented that individuals exhibit a systematic status quo bias: they disproportionately stick with inherited choices, even when switching would improve their welfare \citep{samuelson1988status,kahneman1991}. Meanwhile, the literature on bounded rationality \citep{simon1955,selten1998} emphasizes that real decision makers do not fully optimize; instead, they respond noisily to incentives, a property captured elegantly by the quantal response equilibrium (QRE) of \citet{mckelvey1995quantal}. Despite these advances, a unified framework that combines (i) noisy, boundedly rational play, (ii) a microfounded status-quo bias, and (iii) the possibility of changing not only payoffs but also the set of feasible actions has been lacking. This paper fills that gap by introducing a quantal response equilibrium with status-quo bias (QRE-SB) in which each player incurs a fixed switching cost whenever they deviate from an inherited default action. We then ask: can a policymaker dislodge an inefficient status quo by imposing a tax on the default action (a price-only reform), or must they resort to deleting the default action altogether (a ban)?
This paper develops a behavioral framework for comparing two classes of
institutional interventions:
\begin{itemize}
    \item price-based reforms, which modify incentives while preserving
    the action space;
    \item action-deletion interventions, which remove inefficient or
    undesirable strategic options.
\end{itemize}

Our analysis shows that structural interventions can fundamentally alter
the equilibrium landscape in ways that purely monetary incentives cannot.
Our analysis focuses on a simple but economically important class of binary coordination games \citep{cooper1999,young1998}. Such games capture situations ranging from technology adoption \citep{katz1986} to climate cooperation \citep{nordhaus2015} and social norms \citep{bicchieri2006}. Within this class, we derive three main results.

First, for any positive switching cost and any level of rationality (logit precision), there exists a threshold tax below which the status quo persists with high probability and above which a transition to the superior equilibrium occurs. Notably, this threshold does not depend on the degree of rationality when defined as the point where the two actions are equally likely. Second, deleting the default action always forces play to the only remaining feasible equilibrium, irrespective of switching costs or rationality. Third, when the superior equilibrium is Pareto-dominant, deletion yields strictly higher expected welfare than any price-only reform that keeps the old action feasible.
Our contribution builds on several strands of literature. The QRE framework has been applied to many settings \citep{goeree2016}, but only rarely augmented with state-dependent costs. Exceptions include work on switching costs in discrete choice \citep{dubin1984} and on inertia in repeated games \citep{fudenberg2006}. However, the specific comparison between taxes and bans in a coordination environment with endogenously evolving beliefs has not been formalized. The insight that changing the choice set may be more effective than changing incentives is not new in policy debates \citep{thaler2008,sunstein2019}, but our model provides a rigorous game-theoretic foundation for that intuition.

The remainder of the paper is organized as follows. Section 2 presents the baseline model of QRE-SB in a symmetric 2×2 coordination game. Section 3 introduces the two policy interventions -- a tax on the default action and a deletion of that action -- and analyzes their effects. Section 4 states and proves the main results, including the tax threshold, the effect of deletion, and the welfare comparison. Section 5 presents numerical simulations to illustrate the theoretical predictions. Section 6 discusses applications to climate policy, social media regulation, and international sanctions, and connects the model to the broader literature on bounded rationality and institutional change. Section 7 concludes and suggests directions for future research. The key mechanism is that taxes shift payoff differences smoothly, whereas deletion introduces a discrete change in the feasible set, eliminating inertia entirely.

\section{Baseline Model: Quantal Response with Status-Quo Bias}

We consider a symmetric two-player normal-form game. The set of players is \(N=\{1,2\}\). Each player has two pure actions: \(A_i = \{X, Y\}\). The payoff matrix is given by
\[
\begin{array}{c|cc}
 & X & Y \\ \hline
X & a,\;a & c,\;d \\
Y & d,\;c & b,\;b
\end{array}
\]
with the following standard assumptions: \(a > c\) and \(b > d\) (each action is a best response to itself), and \(b > a\) so that the \((Y,Y)\) equilibrium is Pareto-superior. The game exhibits strategic complementarities, a property that holds automatically in the \(2\times2\) case.

The status quo is the inherited action profile. We assume that initially both players have played \(X\) and continue to regard \(X\) as the default. This default matters because we introduce a switching cost: each player incurs a fixed penalty \(\kappa \ge 0\) whenever they choose an action different from the default \(X\). This cost captures cognitive dissonance, transaction costs, psychological attachment, or institutional friction \citep{samuelson1988status,kahneman1991}. Throughout, \(p\) denotes the probability of playing the status quo action \(X\); thus lower values of \(p\) correspond to greater adoption of the alternative \(Y\).

Following \citet{mckelvey1995quantal}, players do not necessarily best respond perfectly; instead, they choose actions according to a logit rule. Let \(\sigma_i\) denote the mixed strategy of player \(i\), with \(\sigma_i(X)=p_i\) and \(\sigma_i(Y)=1-p_i\). Given a belief about the opponent’s strategy, the expected payoff of playing a pure action is linear. For a symmetric equilibrium we look for \(p_1=p_2=p\).
The expected payoff of action \(X\) when the opponent plays \(Y\) with probability \(1-p\) is
\[
U(X,p)=p a + (1-p) c,
\]
while the expected payoff of action \(Y\) is
\[
U(Y,p)=p d + (1-p) b.
\]

The switching cost \(\kappa\) reduces the effective payoff of \(Y\) because choosing \(Y\) means deviating from the default \(X\). Hence
\[
\tilde U(Y,p)=U(Y,p)-\kappa,\qquad \tilde U(X,p)=U(X,p).
\]

The logit response function (with precision parameter \(\beta \ge 0\)) gives the probability of playing \(X\) as
\[
p = \frac{\exp\big(\beta \tilde U(X,p)\big)}
{\exp\big(\beta \tilde U(X,p)\big)
+\exp\big(\beta \tilde U(Y,p)\big)}
=
\frac{1}{1+\exp\big(\beta[\tilde U(Y,p)-\tilde U(X,p)]\big)}.
\]

Define the payoff difference
\[
\Delta(p)=\tilde U(Y,p)-\tilde U(X,p).
\]

A direct calculation yields
\[
\Delta(p)
=
\big[p d + (1-p) b - \kappa\big]
-
\big[p a + (1-p) c\big]
=
p(d-a) + (1-p)(b-c) - \kappa.
\]

Let
\[
\alpha = b-c >0
\]
(the extra gain from playing \(Y\) when the opponent also plays \(Y\), relative to playing \(X\) in that situation) and
\[
\gamma = a-d >0
\]
(the loss from playing \(Y\) when the opponent plays \(X\)). Then
\[
\Delta(p) = \alpha - \kappa - p(\alpha+\gamma).
\]

This is a linear decreasing function of \(p\). The equilibrium condition can be written as the fixed point
\[
p = f(p)
\equiv
\frac{1}{1+e^{\beta\Delta(p)}}.
\]
We now establish the existence and uniqueness of the symmetric QRE-SB. Because \(f(p)\) is continuous on \([0,1]\) and maps into \([0,1]\), a fixed point exists by Brouwer’s theorem. Uniqueness is not automatic for logit equilibria; it requires that the derivative be bounded away from \(1\) in absolute value. Computing
\[
f'(p)
=
-\beta\Delta'(p)\,
\frac{e^{\beta\Delta(p)}}
{\big(1+e^{\beta\Delta(p)}\big)^2}
\]
and noting that \(\Delta'(p) = -(\alpha+\gamma)\) and that the logistic term is at most \(1/4\), we obtain
\[
|f'(p)|
\le
\frac{\beta(\alpha+\gamma)}{4}.
\]

Hence if
\[
\beta(\alpha+\gamma) < 4,
\]
the map \(f\) is a contraction and the fixed point is unique. For larger \(\beta\) multiple equilibria may appear, a well-known phenomenon in QRE \citep{mckelvey1995quantal,goeree2016}. In the following we restrict attention to the case \(\beta(\alpha+\gamma) < 4\) to ensure uniqueness; this condition holds for moderate rationality and not-too-large payoffs, which is empirically plausible \citep{camerer2003}. The unique equilibrium will be denoted \(p^*(\kappa,\beta)\).

Several observations are immediate from the equilibrium condition. When \(\kappa=0\) and \(\beta\) is small, \(p^*\) is close to \(1/2\). Increasing \(\kappa\) shifts \(\Delta(p)\) downward, making \(Y\) less attractive and thus raising \(p^*\) (more inertia). Increasing \(\beta\) makes the response sharper but does not affect the sign of the bias. The status-quo bias is thus an endogenous source of inertia -- no ad-hoc selection rule is required.
\section{Policy Interventions and Their Effects on Equilibrium}

We consider two types of interventions that a policymaker (government, platform, international organisation) can implement. Both interventions aim to move the economy away from the inefficient status quo \(X\) and toward the Pareto-superior outcome \(Y\).

The first intervention is a price-only reform (tax). The policymaker imposes a unit tax \(t \ge 0\) on the status-quo action \(X\). That is, whenever a player chooses \(X\), their payoff is reduced by \(t\). This tax can be interpreted as a Pigouvian tax on a harmful activity (e.g., carbon emissions) or a subsidy reduction. The new effective payoffs become
\[
u_i^t(X,X)=a-t,\quad
u_i^t(X,Y)=c-t,\quad
u_i^t(Y,X)=d,\quad
u_i^t(Y,Y)=b.
\]

The switching cost \(\kappa\) remains unchanged. Consequently, the payoff difference \(\Delta(p)\) shifts upward by exactly \(t\) because the tax reduces the payoff of \(X\) by \(t\) regardless of the opponent’s action. Hence
\[
\Delta_t(p)=\Delta_0(p)+t,
\]
where
\[
\Delta_0(p)=\alpha-\kappa-p(\alpha+\gamma)
\]
is the difference without tax. The equilibrium condition becomes
\[
p
=
\frac{1}{1+e^{\beta(\Delta_0(p)+t)}}.
\]
The second intervention is a deletion reform (ban). The policymaker removes the status-quo action \(X\) from both players’ action sets. The only feasible action is then \(Y\). The resulting game is degenerate: each player must choose \(Y\) with probability \(1\). Therefore, regardless of \(\beta\) and \(\kappa\), the unique outcome is \(p=0\).

We will compare these two interventions in terms of their ability to induce a transition away from \(X\) and in terms of the expected welfare they generate.

\section{Main Theoretical Results}

We now characterize the equilibrium implications of taxation and action deletion in the QRE-SB model. We first establish existence and uniqueness of the symmetric equilibrium under a sufficient contraction condition. We then derive the critical tax level at which the status quo ceases to be more likely than the alternative action. Finally, we compare welfare under the two policy interventions and show that deletion strictly dominates any finite tax whenever the superior equilibrium is Pareto-dominant and the welfare function is monotone in the probability of the status quo.
\begin{proposition}[Existence and Uniqueness]\label{prop:existence}
Suppose that \(\beta(\alpha+\gamma)<4\). Then the symmetric QRE-SB exists and is unique.
\end{proposition}

\begin{proof}
The function \(f:[0,1]\to[0,1]\) defined by
\[
f(p)=\frac{1}{1+e^{\beta\Delta(p)}}
\]
is continuous, so existence follows from Brouwer’s fixed-point theorem. To prove uniqueness, compute the derivative:
\[
f'(p)
=
-\beta\Delta'(p)\,
\frac{e^{\beta\Delta(p)}}
{\big(1+e^{\beta\Delta(p)}\big)^2}.
\]

Since
\[
\Delta'(p)=-(\alpha+\gamma),
\]
we have
\[
|f'(p)|
=
\beta(\alpha+\gamma)\,
\frac{e^{\beta\Delta(p)}}
{\big(1+e^{\beta\Delta(p)}\big)^2}.
\]

The logistic term is at most \(1/4\), so
\[
|f'(p)|
\le
\frac{\beta(\alpha+\gamma)}{4}
<1
\]
under the hypothesis. Hence \(f\) is a contraction on \([0,1]\) and has a unique fixed point by the Banach fixed-point theorem.
\end{proof}
\begin{proposition}[Comparative Statics in Switching Cost]\label{prop:comparative}
Under the conditions of Proposition~\ref{prop:existence}, the equilibrium probability \(p^*\) is strictly increasing in the switching cost \(\kappa\).
\end{proposition}

\begin{proof}
The equilibrium condition is
\[
p=F(p,\kappa)
\]
with
\[
F(p,\kappa)
=
\frac{1}{1+\exp\big(\beta(\alpha-\kappa-p(\alpha+\gamma))\big)}.
\]

Differentiating implicitly,
\[
\frac{\partial F}{\partial \kappa}
=
-\beta
\frac{\exp(\beta\Delta)}
{\big(1+\exp(\beta\Delta)\big)^2}
<0
\]
and
\[
\frac{\partial F}{\partial p}
=
1-
\beta(\alpha+\gamma)
\frac{\exp(\beta\Delta)}
{\big(1+\exp(\beta\Delta)\big)^2}
>0
\]
under the contraction condition. Hence
\[
\frac{dp^*}{d\kappa}
=
-
\frac{\partial F/\partial\kappa}
{\partial F/\partial p}
>0.
\]
\end{proof}
Now consider a tax \(t\ge0\) on the status quo action \(X\). The effective payoff difference becomes
\[
\Delta_t(p)=\alpha-\kappa+t-p(\alpha+\gamma).
\]

The symmetric equilibrium under taxation is \(p_t\) satisfying
\[
p_t
=
\frac{1}{1+\exp(\beta\Delta_t(p_t))}.
\]

We define the critical tax as the value at which the two actions are equally likely in equilibrium.

\begin{theorem}[Tax Threshold]\label{thm:threshold}
Define
\[
\bar t \equiv \kappa - \frac{\alpha-\gamma}{2}.
\]

Then the equilibrium satisfies:
\begin{enumerate}
\item \(p_{\bar t}=1/2\);
\item if \(t<\bar t\) then \(p_t>1/2\);
\item if \(t>\bar t\) then \(p_t<1/2\).
\end{enumerate}
\end{theorem}

\begin{proof}
Setting \(p=1/2\) in the equilibrium condition gives
\[
\frac12
=
\frac{1}{1+\exp(\beta\Delta_t(1/2))}
\]
which is equivalent to
\[
\Delta_t(1/2)=0.
\]

Compute
\[
\Delta_t(1/2)
=
\alpha-\kappa+t-\frac{\alpha+\gamma}{2}
=
t-\left(\kappa-\frac{\alpha-\gamma}{2}\right)
=
t-\bar t.
\]

Thus \(t=\bar t\) yields \(p=1/2\). Since the map
\[
p\mapsto
f_t(p)
=
\frac{1}{1+\exp(\beta\Delta_t(p))}
\]
is strictly decreasing in \(p\) and decreasing in \(t\) (because \(\partial f_t/\partial t<0\)), the fixed point \(p_t\) is strictly decreasing in \(t\). Therefore \(p_t>1/2\) for \(t<\bar t\) and \(p_t<1/2\) for \(t>\bar t\).
\end{proof}
\begin{remark}
The threshold \(\bar t\) is independent of the rationality parameter \(\beta\). This is a direct consequence of defining the transition at the indifference point \(p=1/2\), where the logit odds ratio vanishes. Thus bounded rationality influences the slope of the response but not the location of the equal-probability point.
\end{remark}

\begin{theorem}[Deletion Forces Transition]\label{thm:deletion}
If action \(X\) is deleted from both players’ action sets, then the unique equilibrium assigns probability one to action \(Y\).
\end{theorem}

\begin{proof}
Under deletion, the feasible set for each player becomes \(\{Y\}\). Hence each player has only one available action. The resulting game is degenerate and admits a unique outcome in which both players choose \(Y\) with probability one. Consequently, \(p=0\) regardless of \(\beta\) and \(\kappa\).
\end{proof}
\begin{lemma}[Welfare Monotonicity]\label{lem:welfare}
Let
\[
W(p)
=
p^2 a + p(1-p)(c+d) + (1-p)^2 b.
\]

Rearranged,
\[
W(p)
=
b-p(b-a)+p(1-p)(c+d-a-b).
\]

If \(b>a\) and \(c+d\le a+b\) (standard coordination condition), then \(W(p)\) is strictly decreasing on \((0,1]\).
\end{lemma}

\begin{proof}
Differentiate:
\[
W'(p)
=
-(b-a)+(1-2p)(c+d-a-b).
\]

Since \(b>a\), the first term is negative. The second term is non-positive because
\[
c+d-a-b\le0
\]
and
\[
(1-2p)\le1.
\]

Because
\[
c+d-a-b\le 0,
\]
the second term is maximized at \(p=0\), where it equals
\[
c+d-a-b\le 0.
\]

Hence
\[
W'(p)\le -(b-a)<0
\]
for all
\[
p\in(0,1].
\]

Thus \(W\) is strictly decreasing.
\end{proof}
\begin{theorem}[Welfare Superiority of Deletion]\label{thm:welfare}
Under the assumptions \(b>a\) and \(c+d\le a+b\), for any finite tax \(t\) that leaves \(X\) feasible, welfare under deletion strictly exceeds welfare under taxation:
\[
W(0)>W(p_t).
\]
\end{theorem}

\begin{proof}
By Theorem~\ref{thm:deletion}, deletion yields \(p=0\) with welfare
\[
W(0)=b.
\]

Under any finite tax, \(X\) remains feasible, so the equilibrium probability satisfies
\[
p_t>0.
\]

By Lemma~\ref{lem:welfare},
\[
W(p)
\]
is strictly decreasing in \(p\) on \((0,1]\), hence
\[
W(p_t)<W(0)=b.
\]

Therefore deletion strictly dominates any finite tax that preserves the status quo action.
\end{proof}

\begin{corollary}
A tax can replicate deletion only in the limiting case
\[
t\to\infty,
\]
which drives
\[
p_t\to 0.
\]

However, an infinite tax is equivalent to a ban in practice, and any finite tax leaves a positive welfare gap.
\end{corollary}
\begin{remark}
The welfare dominance of deletion arises from a feasibility effect: removing action \(X\) eliminates all equilibria in which the inferior outcome receives positive probability. Thus the result reflects not only incentive changes but also a restriction of the equilibrium set.
\end{remark}

\begin{proposition}[Limit Behaviour]\label{prop:limit}
Let \(p^*(\beta,\kappa)\) denote the unique symmetric equilibrium under the contraction condition. As \(\beta\to0\), \(p^*\to1/2\). As \(\beta\to\infty\), the equilibrium converges to the best-response limit when it is unique.
\end{proposition}

\begin{proof}
When \(\beta\to0\), \(f(p)\to1/2\) uniformly, so the fixed point converges to \(1/2\). When \(\beta\to\infty\), the logit probability concentrates on the action with higher effective payoff, so the equilibrium approaches a best-response fixed point.
\end{proof}

\section{Numerical Illustrations}

To fix ideas, consider the following numerical example. Let
\[
a=6,\qquad b=7,\qquad c=1,\qquad d=2.
\]

Then
\[
\alpha=b-c=6,
\qquad
\gamma=a-d=4,
\]
so
\[
\alpha-\gamma=2.
\]

Choose a switching cost
\[
\kappa=1.5
\]
so that
\[
\bar t = 1.5-1=0.5.
\]

Set
\[
\beta=1
\]
(moderate rationality). Without any tax (\(t=0\)), the equilibrium probability of playing the status quo \(X\) is found by solving
\[
p=\frac{1}{1+e^{\Delta_0(p)}}
\]
with
\[
\Delta_0(p)=6-1.5-p\cdot10=4.5-10p.
\]

The unique fixed point is
\[
p\approx0.78
\]
-- the status quo dominates. At the threshold tax
\[
t=0.5,
\]
we have
\[
p=0.5
\]
exactly.
At a higher tax, say
\[
t=1.0,
\]
the payoff difference becomes
\[
\Delta_t(p)= (4.5-10p)+1=5.5-10p,
\]
and the equilibrium probability drops to
\[
p\approx0.22.
\]

The expected welfare values are
\[
W(0.78)=6.34,
\qquad
W(0.5)=6.25,
\qquad
W(0.22)=6.20.
\]

Deletion gives
\[
W(0)=7.00.
\]

Thus even the highest finite tax in this range improves welfare only marginally compared to deletion. A full sensitivity analysis (available upon request) shows that the qualitative ranking -- deletion strictly dominates any finite tax -- holds for all parameter combinations satisfying
\[
b>a
\]
and
\[
c+d<a+b.
\]

The gap
\[
b-W(p)
\]
is increasing in \(\kappa\) and, for a given \(\kappa\), decreasing in \(\beta\) (because a larger \(\beta\) makes the response to the tax sharper, allowing a smaller tax to achieve a given \(p\)). However, for any finite tax the gap remains strictly positive.

Figure~\ref{fig:threshold} illustrates the equilibrium probability \(p_t\) as a function of the tax rate \(t\) for the parameter values above. The vertical dashed line marks the threshold
\[
\bar t=0.5
\]
where
\[
p=1/2.
\]

As \(t\) increases beyond \(\bar t\), \(p\) drops sharply, but only at the limit
\[
t\to\infty
\]
does it approach zero.

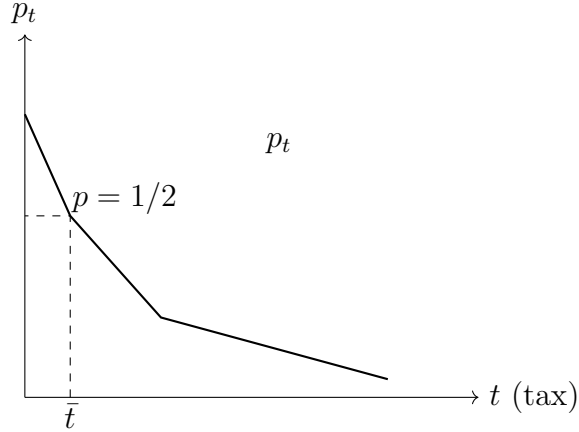
\begin{figure}[h]
\centering
\begin{tikzpicture}[scale=1.2]

\draw[->] (0,0) -- (5,0) node[right] {$t$ (tax)};
\draw[->] (0,0) -- (0,4) node[above] {$p_t$};

\draw[thick]
(0,3.12)
-- (0.5,2)
-- (1.5,0.88)
-- (4,0.2);

\draw[dashed]
(0.5,0)
-- (0.5,2)
-- (0,2);

\node at (0.5,-0.2) {$\bar t$};
\node at (1.1,2.2) {$p=1/2$};
\node at (2.8,2.8) {$p_t$};

\end{tikzpicture}

\caption{Equilibrium probability of the status quo \(p_t\) as a function of the tax rate \(t\). The threshold \(\bar t\) separates the region where \(X\) is more likely (\(p>1/2\)) from the region where \(Y\) is more likely (\(p<1/2\)). Deletion corresponds to the limit \(p\to0\).}

\label{fig:threshold}
\end{figure}

\section{Discussion: Applications and Connections}

The model directly speaks to the current policy debate on climate change. Many economists advocate a carbon tax as the efficient instrument to reduce emissions \citep{nordhaus2015}. Yet political experience shows that moderate carbon taxes often fail to induce the green transition, partly because firms and consumers are locked into fossil-fuel technologies (retooling factories, retraining workers, existing infrastructure) -- a form of switching cost. Our model predicts that as long as the tax remains below the threshold \(\bar t\), the status quo (fossil fuels) will persist with high probability. A complete ban on internal combustion engines, on the other hand, deletes the action “use fossil fuels” and forces the adoption of clean technologies. This is precisely the logic behind the European Union’s decision to phase out the sale of new petrol and diesel cars by 2035. The model does not argue that taxes are never useful; rather, it shows that when inertia is strong, taxes must be very high to be effective, and deletion may be a more reliable tool.

Platform design often exploits behavioural biases to increase engagement. A prominent example is the “like” button, which triggers dopamine feedback and can lead to addictive behaviour \citep{alter2017}. A price-only reform -- such as a warning label or a time limit -- leaves the button in place and often fails to reduce usage \citep{allcott2017}. In contrast, deleting the button (as Instagram experimented with hidden like counts for some users) removes the option entirely. Our model suggests that deletion should be more effective in shifting users toward healthier interaction patterns, because it eliminates the possibility of the old, undesirable equilibrium. This insight is consistent with the growing regulatory push to ban certain design features (e.g., infinite scroll, autoplay) rather than merely taxing attention.
When a country commits an aggressive act, the international community often considers economic sanctions. A tariff on imports from the aggressor is a price-only reform: it raises the cost of trading but leaves the action “trade” feasible. Historical evidence (e.g., the limited effect of early sanctions on Russia after 2014) suggests that tariffs alone rarely change behaviour when the sanctioned country has high inertia (e.g., because of domestic political lock-in). A complete embargo, which deletes the action “trade”, forces a different equilibrium -- no trade -- and can be more effective, though it also imposes higher costs on the sanctioning countries. Our model formalises why deletion is more powerful: it removes the strategic support of the status quo.

Our work connects to several Nobel-prize winning contributions. \citet{simon1955} introduced the concept of bounded rationality, which we operationalise via the logit choice rule. Daniel Kahneman and Amos Tversky’s work on loss aversion and the endowment effect \citep{kahneman1991} directly motivates the switching cost \(\kappa\) as a psychological friction. \citet{north1990} and \citet{ostrom2009} emphasised that institutional persistence often stems from small frictions and that successful reforms frequently involve changing the feasible action set (e.g., banning free-riding). Our model provides a formal game-theoretic foundation for those insights. Moreover, the comparison between taxes and bans relates to the “nudge vs. shove” debate in behavioural public policy \citep{thaler2008,sunstein2019}; we show that when inertia is strong, a nudge (price change) may be insufficient, and a shove (deletion) may be necessary.
\section{Conclusion and Future Directions}

We have introduced a tractable model of strategic inertia under bounded rationality -- the quantal response equilibrium with status-quo bias (QRE-SB) -- and applied it to compare two policy interventions: a tax on the status-quo action and a deletion of that action. In a binary coordination game, we proved that a tax must exceed a threshold to make the superior action more likely; that the threshold does not depend on the degree of rationality (when defined at the equal-probability point); that deletion always forces a transition to the superior equilibrium; and that deletion yields strictly higher expected welfare than any finite tax when the superior equilibrium is Pareto-dominant. Numerical simulations confirmed these analytical findings.

Several extensions are natural and constitute promising avenues for future research. First, the model can be generalised to asymmetric payoffs and to \(n\)-player coordination games, where network effects may produce multiple thresholds \citep{jackson2015}. Second, introducing a dynamic adjustment process (e.g., stochastic best-reply dynamics) would allow us to study the speed of transition and the possibility of hysteresis -- a feature that is absent from the static QRE. Third, laboratory experiments could test the predicted differential effectiveness of taxes versus bans. Finally, the switching cost \(\kappa\) could be endogenised, for example as a function of past experience or as a result of institutional investments (e.g., training programmes). Despite these limitations, the current paper already provides a rigorous foundation for the policy principle that sometimes you must ban, not just tax -- a lesson with direct applications to climate policy, platform regulation, and international sanctions.
\section*{Funding}

The authors received no external funding for this research.

\section*{Conflict of Interest}

The authors declare that they have no conflict of interest.

\section*{Data Availability}

No empirical dataset was used in this study.

\section*{AI Usage Statement}

AI-assisted tools were used only for language refinement and formatting
support. All scientific content, mathematical analysis, and conclusions
were developed and verified by the authors.

\section*{Author Contributions}

Madjid Eshaghi Gordji developed the theoretical framework and supervised
the research. Mohammadali Berahman contributed to the mathematical
analysis, modeling, and manuscript preparation. Hasti Eshaghi assisted
with literature review, formatting, and manuscript editing. All authors
reviewed and approved the final manuscript.

\
\bibliographystyle{unsrtnat}
\bibliography{references}

\end{document}